\documentclass[11pt]{article}
\usepackage{graphicx}
\usepackage{fix-cm}
\usepackage{mathptmx}      
\usepackage{lp,graphicx,etoolbox}
\usepackage{amsmath, amssymb}
\usepackage{subfig,url}

\usepackage[noline,noend,linesnumberedhidden,boxed]{algorithm2e}
\DontPrintSemicolon
\SetArgSty{textnormal}
\SetAlCapSkip{1ex}
\SetAlgoCaptionLayout{centerline}

\newcommand{\qed}{\rule{7pt}{7pt}}

\newenvironment{proof}{\noindent {\it Proof\/}:}{$\qed$ \medskip}
 \newcommand{\set}[1]{\{ #1 \}}
\newtheorem{theorem}{Theorem}[section]
\newtheorem{lemma}[theorem]{Lemma}
\newtheorem{corollary}[theorem]{Corollary}

\newtheorem{remark}[theorem]{Remark}

\def\({\left(}
\def\){\right)}

\newcommand{\mqed}{}
\newlength{\fixSUBT}
\newcommand{\SUBT}{\ensuremath{SU\hspace*{-\fixSUBT}BT}}

\newtoggle{abs}
\toggletrue{abs}

\begin{document}

\title{On the Integrality Gap of the Subtour LP for the 1,2-TSP\thanks{A preliminary version of this paper appeared in LATIN 2012: Theoretical Informatics \cite{QianSWvZ11}. Some results in the current paper are stronger, using the same techniques as in~\cite{QianSWvZ11}.
The first author was supported in part by NSF grant CCF-1115256.  This work was carried out in part while the third author was on sabbatical at TU Berlin, and the fourth author was at Max-Planck-Institut f\"ur Informatik in Saarbr\"ucken, Germany.
The third author was supported in part by the Berlin Mathematical School, the Alexander von Humboldt Foundation, and NSF grant CCF-1115256.}
}


\author{Jiawei Qian\thanks{Address: Bank of America Merrill Lynch,
Chicago, IL, 60661, USA.   Email: {\tt jq35@cornell.edu}.  Supported in part by NSF grant CCF-0830519.}
\and Frans Schalekamp\thanks{Address: Department of Mathematics, College of William and Mary, Williamsburg, VA 23185, USA. Email:{\tt frans@wm.edu}}
\and David P.\ Williamson\thanks{Address: School
of Operations Research and Information Engineering, Cornell
University, Ithaca, NY 14853, USA.  Email: {\tt dpw@cs.cornell.edu}.  This work was carried out while the author was on sabbatical at TU Berlin.
Supported in part by the Berlin Mathematical School, the Alexander von Humboldt Foundation, and NSF grant CCF-1115256.}
\and Anke van Zuylen
\thanks{Address: Department of Mathematics, College of William and Mary, Williamsburg, VA 23185, USA. Email:{\tt anke@wm.edu}}}

\date{}
\maketitle

\newtheorem{conj}{Conjecture}

\begin{abstract}
In this paper, we study the integrality gap of the subtour LP relaxation for the traveling salesman problem in the special case when all edge costs are either 1 or 2. For the general case of symmetric costs that obey triangle inequality, a famous conjecture is that the integrality gap is 4/3. Little progress towards resolving this conjecture has been made in thirty years.
We conjecture that when all edge costs $c_{ij}\in \{1,2\}$, the integrality gap is $10/9$. We show that this conjecture is true when the optimal subtour LP solution has a certain structure. Under a weaker assumption, which is an analog of a recent conjecture by Schalekamp, Williamson and van Zuylen, we show that the integrality gap is at most $7/6$. When we do not make any assumptions on the structure of the optimal subtour LP solution, we can show that integrality gap is at most $5/4$; this is the first bound on the integrality gap of the subtour LP strictly less than $4/3$ known for an interesting special case of the TSP.  We show computationally that the integrality gap is at most $10/9$ for all instances with at most 12 cities.
\end{abstract}

\section{Introduction}
The Traveling Salesman Problem (TSP) is one of the most well studied problems in combinatorial optimization. Given a set of cities $\{1, 2, \ldots, n\}$, and distances $c( i, j )$ for traveling from city $i$ to $j$, the goal is to find a tour of minimum length that visits each city exactly once.
An important special case of the TSP is the case when the distance forms a metric, i.e., $c(i,j)\le c(i,k) + c(k,j)$ for all $i,j,k$, and all distances are symmetric, i.e., $c(i,j)=c(j,i)$ for all $i,j$. The symmetric TSP is known to be NP-hard, even if $c(i,j)\in \{1,2\}$ for all $i,j$~\cite{PapadimitriouY93}; note that such instances trivially obey the triangle inequality.  Such instances are also known to be APX-hard; that is, there is no $\alpha$-approximation algorithm for the problem for some $\alpha > 1$ unless $P = NP$.

The metric TSP can be approximated to within a factor of $\tfrac 32$ using an algorithm by Christofides~\cite{Christofides76} from 1976.  The algorithm combines a minimum spanning tree with a matching on the odd-degree nodes to get an Eulerian graph that can be shortcut to a tour; the analysis shows that the minimum spanning tree and the matching cost no more than the optimal tour and half the optimal tour respectively.
Better results are known for several special cases, but, surprisingly, no progress has been made on approximating the general symmetric TSP in more than thirty years. A natural direction for trying to obtain better approximation algorithms is to use linear programming. The following linear programming relaxation of the traveling salesman problem was used by Dantzig, Fulkerson, and Johnson~\cite{DantzigFJ54} in 1954.  For simplicity of notation, we let $G=(V,E)$ be a complete undirected graph on $n$ nodes.  In the LP relaxation, we have a variable $x(e)$ for all $e = (i,j)$ that denotes whether we travel directly between cities $i$ and $j$ on our tour.  Let $c(e) = c(i,j)$, and let $\delta(S)$ denote the set of all edges with exactly one endpoint in $S \subseteq V$.
Then the relaxation is \lps & &
& \mbox{Min} & \sum_{e \in E} c(e) x(e) \\
(\SUBT) & \mbox{subject to:} & & & \sum_{e \in \delta(i)} x(e) = 2, & \forall i \in V, \numb{degreecons} \\
& & & & \sum_{e \in \delta(S)}x(e) \geq 2, & \forall S\subset V,\, 3 \leq |S| \leq |V|-3, \numb{subtourcons}\\
& & & & 0 \leq x(e) \leq 1, & \forall e \in E. \numb{boundscons} \elps The first
set of constraints (\ref{degreecons}) are called the {\em degree constraints}.  The second set of constraints (\ref{subtourcons}) are sometimes called {\em subtour elimination constraints} or sometimes just {\em subtour constraints}, since they prevent solutions in which there is a subtour of just the nodes in $S$. As a result, the linear program is sometimes called the {\em subtour LP}.
It has been shown by Wolsey~\cite{Wolsey80} (and later Shmoys and Williamson~\cite{ShmoysW90}) that Christofides' algorithm finds a tour of length at most~$\tfrac 32$ times the optimal value of the subtour LP; these proofs show that the minimum spanning tree and the matching on odd-degree nodes can be bounded above by the value of the subtour LP, and half the value of the subtour LP, respectively.  This implies that the integrality gap, the worst case ratio of the length of an optimal tour divided by the optimal value of the LP, is at most~$\tfrac 32$. However, no examples are known that show that the integrality gap can be as large as~$\tfrac 32$; in fact, no examples are known for which the integrality gap is greater than $\tfrac 43$. A well known conjecture states that the integrality gap is indeed $\tfrac 43$; see (for example) Goemans~\cite{Goemans95}.

Recently, progress has been made in several directions, both in improving the best approximation guarantee and in determining the exact integrality gap of the subtour LP for certain special cases of the symmetric TSP.
In the {\em graph}-TSP, the costs $c(i,j)$ are equal to the shortest path distance in an underlying unweighted graph.  If the graph is cubic and 3-connected, Gamarnik, Lewenstein and Sviridenko~\cite{GamarnikLS05} show an approximation algorithm with guarantee slightly better than $\tfrac 32$.
Oveis Gharan, Saberi, and Singh~\cite{OveisGharanSS10} show that the graph-TSP can be approximated to within $\tfrac 32-\epsilon$ for a small constant $\epsilon>0$ for all graphs.
Boyd, Sitters, van der Ster and Stougie~\cite{BoydSSS11}, and Aggarwal, Garg and Gupta~\cite{AggarwalGG11} independently give a $\tfrac 43$-approximation algorithm if the underlying graph is cubic.
M\"omke and Svensson~\cite{MomkeS11} improve these results by giving a 1.461-approximation for the graph-TSP and an $\tfrac 43$-approximation algorithm if the underlying graph is subcubic. Mucha~\cite{Mucha11} improves the analysis of the M\"omke-Svensson algorithm to a  $\tfrac{13}{9}$-approximation algorithm, and Seb\H o and Vygen~\cite{SeboV12} combine the ideas of M\"omke and Svensson~\cite{MomkeS11} with an algorithm based on a carefully chosen ear decomposition of the graph to obtain a $\tfrac75$-approximation algorithm.  All of these $\alpha$-approximation algorithms imply a corresponding upper bound of $\alpha$ on the integrality gap for the corresponding graph-TSP instances.

In Schalekamp, Williamson and van Zuylen~\cite{SchalekampWvZ11}, three of the authors of this paper resolve a related conjecture. A 2-matching of a graph is a set of edges such that no edge appears twice and each node has degree two, i.e., it is an integer solution to the LP $(\SUBT)$  with only constraints (\ref{degreecons}) and (\ref{boundscons}). Note that a minimum-cost 2-matching thus provides a lower bound on the length of the optimal TSP tour. A minimum-cost 2-matching can be found in polynomial time using a reduction to a certain minimum-cost matching problem. Boyd and Carr~\cite{BoydC11} conjecture that the worst case ratio of the cost of a minimum-cost 2-matching and the optimal value of the subtour LP is at most $\tfrac {10}9$. This conjecture was proved to be true by Schalekamp et al.\ and examples are known that show this result is tight.

Unlike the techniques used to obtain better results for the graph-TSP, the techniques of Schalekamp et al.\ work on general weighted instances that are symmetric and obey the triangle inequality. However, their results only apply to 2-matchings and it is not clear how to enforce global connectivity on the solution obtained by their method. A potential direction for progress on resolving the integrality gap for the subtour LP is a conjecture by Schalekamp et al.\ that the worst-case integrality gap is attained for instances for which the optimal subtour LP solution is a basic solution to the linear program obtained by dropping the subtour elimination constraints.

In this paper, we turn our attention to the 1,2-TSP, where $c(i,j)\in\{1,2\}$ for all $i,j$. Note that bounding the cost of enforcing connectivity is relatively easy in this class of problems, since we may connect any two components for an increase in cost of at most 2. Papadimitriou and Yannakakis~\cite{PapadimitriouY93} show how to approximate 1,2-TSP within a factor of $\tfrac{11}9$ by computing a minimum-cost 2-matching and merging its cycles into a tour. In addition, they show a ratio of $\tfrac 76$ if they start with a minimum-cost 2-matching that has no cycles of length 3. Bl\"aser and Ram~\cite{BlaeserR05} improve this ratio and the best known approximation factor of $\tfrac 87$ is given by Berman and Karpinski~\cite{BermanK06}.

We do not know a tight bound on the integrality gap of the subtour LP even in the case of the 1,2-TSP.  As an upper bound, we appear to know only that the gap is at most $\tfrac 32$ via Wolsey's result.   There is an easy 9 city example showing that the gap must be at least $\tfrac{10}9$; see Figure~\ref{fig:12tsp}.  This example has been extended to a class of instances on $9k$ nodes for any positive integer $k$ by Williamson~\cite{Williamson90}.  The contribution of this paper is to  begin a study of the integrality gap of the 1,2-TSP, and to improve our state of knowledge for the subtour LP in this case. We prove an upper bound on the integrality gap for the subtour LP of $\tfrac 54$, which is the first bound on the integrality gap with value less than $\tfrac 43$ for a natural class of TSP instances. Under an analog of a conjecture of Schalekamp et al.~\cite{SchalekampWvZ11}, we show that the integrality gap is at most $\tfrac 76$, and with an additional assumption on the structure of the solution, we can improve this bound to $\tfrac{10}9$. We describe these results in more detail below.

\begin{figure}[t]
\begin{center}
\includegraphics[height=.75in]{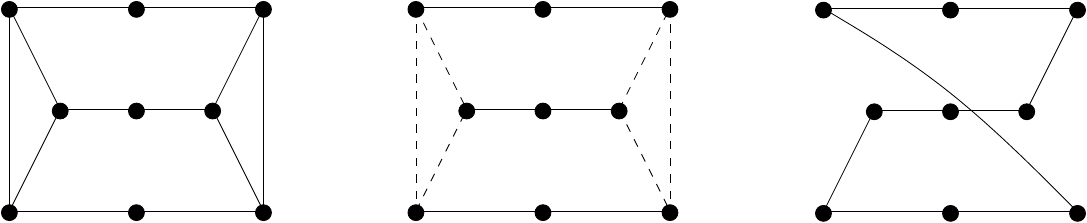}
\end{center}
\caption{Illustration of the worst example known for the integrality gap for
the 1,2-TSP.  The figure on the
left shows all edges of cost 1.  The figure in the center gives the subtour LP
solution, in which the dotted edges have value $\tfrac 12$, and the solid
edges have value 1; this is also an optimal fractional 2-matching.  The figure on the left gives the optimal
tour and the optimal 2-matching.} \label{fig:12tsp}
\end{figure}

All the known approximation algorithms since the initial work of Papadimitriou and Yannakakis~\cite{PapadimitriouY93} on the problem start by computing a minimum-cost 2-matching.
However, the example of Figure~\ref{fig:12tsp} shows that an optimal 2-matching can be as much as $\tfrac{10}9$ times the value of the subtour LP for the 1,2-TSP, so we cannot directly replace the bound on the optimal solution in these approximation algorithms with the subtour LP in the same way that Wolsey did with Christofides' algorithm in the general case. Using the result of Schalekamp, Williamson, and van Zuylen~\cite{SchalekampWvZ11} and a new lemma that relates part of the analysis of Papadimitriou and Yannakakis~\cite{PapadimitriouY93} to the subtour LP bound, we obtain a preliminary upper bound on the integrality gap of the subtour LP for the 1,2-TSP of $\tfrac{7}{9}\cdot\tfrac{10}{9}+\tfrac{4}{9} = \tfrac{106}{81} \approx 1.3086.$ 

To improve this upper bound to $\tfrac{5}{4}$, we first show stronger results in some cases.
A fractional 2-matching is a basic optimal solution to the LP $(\SUBT)$ with only constraints (\ref{degreecons}) and (\ref{boundscons}).
Schalekamp et al.~\cite{SchalekampWvZ11} have conjectured that the worst-case integrality gap for the subtour LP is obtained when the optimal solution to the subtour LP is an extreme point of the fractional 2-matching polytope.
We show that if this is the case for 1,2-TSP
then we can find a tour of cost at most $\tfrac 76$ the cost of the fractional 2-matching, implying that the integrality gap is at most $\tfrac 76$ in these cases.
Next, we show that if this optimal solution to the fractional 2-matching problem has a certain structure, then we can find a tour of cost at most $\tfrac{10}9$ times the cost of the fractional 2-matching, implying an upper bound on the integrality gap of $\tfrac{10}9$ for these cases.  Figure~\ref{fig:12tsp} shows that this result is tight.  

We then use the previous arguments to show that one can construct a tour of cost at most $\tfrac{5}{4}$ times the subtour LP value. To do this, we prove that we can assume without loss of generality that the optimal value of the subtour LP is less than $n+1$, where $n$ denotes the number of nodes. Combined with a more careful analysis based on the results obtained before, we obtain our main result.
The results above all lead to polynomial-time algorithms, though we do not state the exact running times.

Finally, we perform computational experiments to show that the integrality gap is at most $\tfrac{10}9$ for $n \leq 12$.  We conjecture that the integrality gap is in fact exactly $\tfrac{10}9$.

We note that the upper bound on the integrality gap for general 1,2-TSP instances presented in this paper is stronger than the bound that appeared in a preliminary version of this paper~\cite{QianSWvZ11} of $\tfrac{19}{15}$. In the time between publication of the preliminary version and the current revision, Mnich and M\"omke~\cite{MnichM13} obtained an upper bound of $\tfrac{5}{4}$ on the integrality gap for 1,2-TSP instances that have the additional property of being ``fractionally Hamiltonian'', which means that the optimal objective value of the subtour LP is equal to the number of nodes in the instance. In this version, using the same techniques as in the preliminary version, we show an unconditional upper bound on the integrality gap of $\tfrac{5}{4}$, and a bound of $\tfrac{26}{21}$ for fractionally Hamiltonian instances.

The remainder of this paper is structured as follows.  Section~\ref{sec:12gap} contains preliminaries and a first general bound on the integrality gap for the 1,2-TSP.  We show how to obtain stronger bounds if the optimal subtour LP solution is a fractional 2-matching in Section~\ref{sec:12-2m}. In Section~\ref{sec:better}, we combine the arguments from the previous sections and show that the integrality gap without any assumptions on the structure of the subtour LP solution is at most $\tfrac{5}{4}$. We describe our computational experiments in Section~\ref{sec:comp}.  Finally, we close with a conjecture on the integrality gap of the subtour LP for the 1,2-TSP in Section~\ref{sec:conc}.
\iftoggle{abs}{Some proofs are omitted due to space reasons and can be found in the full version. of the paper.}{}

\section{Preliminaries and a first bound on the integrality gap}
\label{sec:12gap}
We will work extensively with 2-matchings and fractional 2-matchings; that is, extreme points $x$ of the LP $(\SUBT)$ with only constraints (\ref{degreecons}) and (\ref{boundscons}), where in the first case the solutions are required to be integer. For convenience we will abbreviate ``fractional 2-matching'' by F2M and ``2-matching'' by 2M.  The basic solutions of the F2M polytope have the following well-known structure (attributed to Balinski~\cite{Balinski65}).  Each connected component of the support graph (that is, the edges $e$ for which $x(e) > 0$) is either a cycle on at least three nodes with $x(e)=1$ for all edges $e$ in the cycle, or consists of odd-sized cycles with $x(e) = \tfrac 12$ for all edges $e$ in the cycle connected by paths of edges $e$ with $x(e) = 1$ for each edge $e$ in the path (the center figure in Figure~\ref{fig:12tsp} is an example).  We call the former components {\em integer components} and the latter {\em fractional components}.  In a fractional component, we call a path of edges $e$ with $x(e)=1$ a {\em 1-path}.  The edges $e$ with $x(e)=\tfrac 12$ in cycles are called {\em cycle edges}.  An F2M with a single component is called {\em connected}, and we call a component {\em 2-connected} if the sum of the $x$-values on the edges crossing any cut is at least $2$. We let $n$ denote the number of nodes in an instance.

As mentioned in the introduction, Schalekamp, Williamson, and van Zuylen~\cite{SchalekampWvZ11} have  shown the following.

\begin{theorem}[Schalekamp et al.~\cite{SchalekampWvZ11}]
\label{thm:109}
If edge costs obey the triangle inequality, then the cost of an optimal 2-matching is at most $\tfrac{10}9$ times the value of the subtour LP.
\end{theorem}
It is not hard to show that this immediately implies an upper bound of $\tfrac 43 \times \tfrac{10}9$ on the integrality gap of the subtour LP for the 1,2-TSP: we can just compute a minimum cost 2-matching at a cost of $\tfrac{10}9$ the value of the subtour LP, remove the most expensive edge from each cycle, which gives a collection of node-disjoint paths, and add edges of cost 2 to combine these paths into a tour.  Each cycle has at least 3 edges; at worst, we remove an edge of cost 1 from each cycle and then need an edge of cost 2 to patch the paths into a tour.  Thus the overall cost increases by a factor of $\tfrac43$, giving a tour of cost at most $\tfrac43 \times \tfrac{10}9$ times the value of the subtour LP.

The algorithm of Papadimitriou and Yannakakis~\cite{PapadimitriouY93} improves on this idea, by cleverly merging the cycles of the optimal 2M solution.
We summarize the properties of their algorithm that we will use. 
First, observe that we can assume without loss of generality that the optimal 2M solution consists of a number of cycles with only edges of cost 1 (``pure'' cycles) and at most one cycle which has one or more edges of cost 2 (the ``non-pure'' cycle), by deleting the edges of cost 2 and combining the resulting disjoint paths into a single cycle.  Moreover, if $i$ is a node in the non-pure cycle which is incident on an edge of cost 2 in the cycle, then there can be no edge of cost 1 connecting $i$ to a node in a pure cycle (since otherwise, we can merge the non-pure cycle with a pure cycle without increasing the cost).

The Papadimitriou-Yannakakis algorithm solves the following bipartite matching problem: On one side we have a node for every pure cycle, and on the other side, we have a node for every node in the instance. There is an edge from pure cycle $C$ to node $i$, if $i\not\in C$ and there is an edge of cost 1 from $i$ to some node in $C$. Let $r$ be the number of pure cycles that are unmatched in a maximum cardinality bipartite matching. Papadimitriou and Yannakakis show how to ``patch together'' the matched cycles. We refer the reader to their original paper~\cite{PapadimitriouY93} for more details.
The resulting cycles are then combined into a tour of cost at most 
\begin{equation}\tfrac {7}9 OPT(2M) + \tfrac 49 n + \tfrac 13 r,\label{eq:PY}\end{equation}
where $OPT(2M)$ is the cost of an optimal 2M solution.\footnote{In~\cite{PapadimitriouY93}, $OPT(2M)$ is expressed as $n+k$, where $k$ is the number of edges of cost 2 in the optimal 2M solution. The number of unmatched pure cycles is denoted by $r_2$. The bound given by~\cite{PapadimitriouY93} is $n+k +\tfrac29(n-n_2-k) + r_2$, where $n_2$ is a quantity that is lower bounded by $3r_2$. Therefore, the bound in~\cite{PapadimitriouY93} can be upper bounded by $\tfrac79(n+k)+\tfrac 49n + \tfrac13 r_2$.}

We now show how to convert this bound into a bound in terms of the optimal value to \SUBT.
\begin{lemma}\label{lem:PYsubt}
Let $r$ be the number of pure cycles that are unmatched in a maximum cardinality bipartite matching instance defined by Papadimitriou and Yannakakis.
Then 
\[OPT(\SUBT) \ge n + r.\]
\end{lemma}

\begin{proof}
We note that for a bipartite matching instance, the size of the minimum cardinality vertex cover is equal to the size of the maximum matching. We use this fact to construct a feasible dual solution to the subtour LP that has value $n+r$.
Let ${\cal C}_M, V_M$ be the pure cycles and nodes (in the original graph), for which the corresponding nodes in the bipartite matching instance are in the minimum cardinality vertex cover.
The dual of the subtour LP $(\SUBT)$ is \lps & &
& \mbox{Max} & 2\sum_{S \subset V} y(S) + 2\sum_{i \in V} y(i) - \sum_{e \in E} z(e)\\
(D) &\mbox{subject to:} 
& & & \sum_{S \subset V: e \in \delta(S)} y(S) + y(i) + y(j) - z(e) \leq c(e), & \forall e = (i,j), \\
& & & & y(S) \geq 0, &\hspace*{-2.25cm} \forall S\subset V,\, 3 \leq |S| \leq n-3,\\
& & & & z(e) \geq 0, &\hspace*{-2.25cm} \forall e \in E. \elps
We set $z(e) = 0$ for each $e \in E$, and we set $y(i) =\tfrac 12$ for each $i\in V\backslash V_M$. For a pure cycle on a set of nodes $C$, we set $y(C)=\tfrac12$, if the cycle is not in ${\cal C}_M$. The dual objective for this solution is exactly $n+r$: its value is $n$ plus the number of pure cycles minus the size of the vertex cover, or $n$ plus the number of pure cycles minus the size of the matching, since the vertex cover has the same size as the matching.  Thus it is the same as $n$ plus the number of pure cycles not in the matching, or $n+r$.

It remains to show that the dual constructed is feasible.  Define the {\em load} on an edge $e = (i,j)$ of solution $(y,z)$ to be $\sum_{S \subset V: e \in \delta(S)} y(S) + y(i) + y(j) - z(e)$.  For any edge $e=(i,j)$  of cost 1 inside a cycle of the 2M, the load on the edge is at most 1, since the only potentially non-zero dual variables loading the edge are the dual variables $y(i)$ and $y(j)$. For an edge $(i,j)$ where $i\in C$ and $j\in C'\neq C$, the load is $y(i)+y(j)+y(C)+y(C')\le 2$.
Suppose $(i,j)$ has cost 1, and the cycles $C$ and $C'$ are both pure cycles. Then the edge occurs twice in the bipartite matching instance (namely, once going from $i$ to $C$ and once going from $j$ to $C'$) and hence the dual of at least two of the four objects $i,j, C$ and $C'$ has been reduced to 0. The total load on edge $(i,j)$ is thus at most 1.
Now, suppose $C'$ is the non-pure cycle, then $y_{C'}=0$, since we only increased the dual variables for the pure cycles. Moreover, at least one endpoint of the $(j,C)$ edge in the bipartite matching instance must be in the vertex cover, so the load on edge $(i,j)$ is again at most 1.
\end{proof}

Note that, combined with (\ref{eq:PY}) and Theorem \ref{thm:109}, Lemma~\ref{lem:PYsubt} implies that the cost of the tour is at most $ \tfrac{7}{9} \cdot \tfrac{10}{9} OPT(\SUBT) + \tfrac{4}{9} OPT(\SUBT)=\tfrac {106}{81}OPT(\SUBT).$
This bound obtained on the integrality gap seems rather weak, as the best known lower bound on the integrality gap is only $\tfrac{10}9$.
Schalekamp, Williamson, and van Zuylen~\cite{SchalekampWvZ11} have conjectured that the integrality gap (or worst-case ratio) of the subtour LP occurs when the solution to the subtour LP is a fractional 2-matching.

\begin{conj}[Schalekamp et al.~\cite{SchalekampWvZ11}] \label{conj:tsp}
Let $\alpha_n$ be the integrality gap of the subtour LP on $n$ vertices.  Then there exists an instance which has an optimal subtour LP solution that is an F2M and for which the optimal tour has cost at least $\alpha_n$ times the subtour LP cost.
\end{conj}

In the next section, we show that we can obtain better bounds on the integrality gap of the subtour LP in the case that the optimal solution is a fractional 2-matching.
In Section~\ref{sec:better}, we then show how to combine Lemma~\ref{lem:PYsubt} with the bounds in the next section to obtain an upper bound of $\tfrac{5}{4}$ on the integrality gap.

\section{Better bounds if the optimal solution is an F2M}
\label{sec:12-2m}
If the optimal solution to the subtour LP is a fractional 2-matching, then a natural approach to obtaining a good tour is to start with the edges with $x$-value 1, and add as many edges of cost 1 and $x$-value $\tfrac 12$ as possible, without creating a cycle on a subset of the nodes. In other words, we will propose an algorithm that creates an acyclic spanning subgraph $(V,T)$ where all nodes have degree one or two.
We will refer to an acyclic spanning subgraph in which all nodes have degree one or two as a partial tour.
A partial tour can be extended to a tour by adding $d/2$ edges of cost $2$, where $d$ is the number of degree $1$ nodes. The cost of the tour is $c( T ) + d$, where $c( T ) = \sum_{e \in T} c(e)$.

\iftoggle{abs}{}{We will use the following lemma.}
\begin{lemma}\label{lem:third}
Let $G=(V,T)$ be a partial tour.
Let $A$ be a set of edges not in $T$ that form an odd cycle or a path on $V'\subset V$, where the nodes in $V'$ have degree one in $T$.
We can find $A'\subset A$ such that $(V,T\cup A')$ is a partial tour, and
\begin{itemize}
\item
 $|A'|\ge \tfrac 13 |A|$ if $A$ is a cycle,
 \item
 $|A'|\ge \tfrac 13 (|A|-1)$ if $A$ is a path,
 \end{itemize}
\end{lemma}

\iftoggle{abs}{We will now use the lemma above to show a bound of $\tfrac 76$ on the integrality gap if the optimal subtour LP solution is a fractional 2-matching.}{We postpone the proof of the lemma and first prove the implication for the bound on the integrality gap if the optimal subtour LP solution is a fractional 2-matching.}

\begin{theorem}\label{thm:tour76}
There exists a tour of cost at most $\tfrac76$ times the cost of a connected F2M solution if $c(i,j)\in \{1,2\}$ for all $i,j$.
\end{theorem}
\begin{proof}
Let $P = \{ e \in E : x(e) = 1 \}$ (the edges in the 1-paths of $x$).
We will start the algorithm with $T = P$. Let $R = \{ e \in E : x(e) = \tfrac 12 \mbox{ and } c(e) = 1 \}$ (the edges of cost $1$ in the cycles of $x$). 
Note that the connected components of the graph $(V,R)$ consist of paths and odd cycles. The main idea is that we consider these components one by one, and use Lemma~\ref{lem:third} to show that we can add a large number of the edges of each path and cycle, where we keep as an invariant that $T$ is a partial tour.
Note that by Lemma~\ref{lem:third}, the number of edges added from each path or cycle $A$ is at least $|A|/3$, except for the paths for which $|A| \equiv 1 \pmod 3$. Let ${\cal P}_1$ be this set of paths.
We would like to claim that we add a third of the edges on average from each component, and we therefore preprocess the paths in ${\cal P}_1$, where we add one edge (either the first or last edge from each path in ${\cal P}_1$) to $T$ if this is possible without creating a cycle in $T$, and if so, we remove this edge and its neighboring edge in $R$ (if any) from $R$. After the preprocessing, we use Lemma~\ref{lem:third} to process each of the components in $(V,R)$.

We call a path $A$ in ${\cal P}_1$ ``eared'' if the 1-paths that are incident on the first and last node of the path are such that they go between two neighboring nodes of $A$. \iftoggle{abs}{It is not hard to show that we can add an edge from at least half of the paths in ${\cal P}_1$ that are not eared.}{It is not hard to see that we can add an edge from at least half of the paths in ${\cal P}_1$ that are not eared:
If we cannot add either the first or the last edge from a path $A$ in ${\cal P}_1$, and $A$ is not eared, then it must be the case that either the first or the last edge forms a cycle with an edge that was added to $T$ from another path in ${\cal P}_1$. 
}

\iftoggle{abs}{}{After preprocessing the paths in ${\cal P}_1$, we iterate through the connected components in $(V,R)$ and add edges to $T$ while maintaining that $T$ is a partial tour.
By Lemma~\ref{lem:third}, the number of edges added from each path or cycle $A$
is at least $|A|/3$, except for the paths in ${\cal P}_1$. }
We now consider two cases for the paths in ${\cal P}_1$, depending on whether we added an edge from the path to $T$ in the preprocessing step or not. 
Note that for a path $A$ in ${\cal P}_1$
for which we added an edge to $T$ in the preprocessing step, $R$ contains a path of $|A|-2$ edges after the preprocessing step, and by Lemma~\ref{lem:third}, we add at least $(|A|-2-1)/3$ of these to $T$. Together with the edge added in the preprocessing step, we thus add at least $1+(|A|-2-1)/3 = |A|/3$ edges. 
For a path in ${\cal P}_1$ for which we did not add an edge to $T$ in the preprocessing stage, we add at least $ (|A|-1)/3$ edges.
Now, recall that a path $A$ in ${\cal P}_1$ has $|A|\equiv 1\pmod 3$, and that the number of edges added is an integer, so in the first case, the number of edges added is at least $|A|/3 + \tfrac 23$ and in the second case it is $|A|/3-\tfrac 13$.
Let $z$ be the number of eared paths in ${\cal P}_1$. Then, the number of paths in ${\cal P}_1$ that are in the second case is at most $z$ plus the number of paths in ${\cal P}_1$ that fall in the first case. Hence, the total number of edges from $R$ that were added to $T$ can be lower bounded by $\tfrac 13 |R|-\tfrac 13 z$. We now give an upper bound on the number of nodes of degree one in $T$.

Let $k$ be the number of cycle nodes in $x$, i.e. $k = \# \{ i \in V : x(i,j) = \tfrac 12\mbox{ for some $j\in V$} \}$, and let $p$
be the number of cycle edges of cost 2 in $x$, i.e. $p = \# \{ e \in E : x(e) = \tfrac 12 \mbox{ and } c(e) = 2 \} $.
Note that $(V,R)$ contains $p$ paths (which may have zero edges) on the cycle nodes, and hence $p\ge z$.
Initially, when $T$ contains only the edges in the 1-paths, all $k$ nodes have degree one, and there are $k-p$ edges in $R$. We argued that we added at least $\tfrac 13 |R|-\tfrac 13 z = \tfrac 13 k - \tfrac 13 p-\tfrac 13 z$ edges to $T$. Each edge reduces the number of nodes of degree one by two, and hence, the number of nodes of degree one at the end of the algorithm is at most $k - 2(\tfrac 13 k - \tfrac 13 p - \tfrac 13 z) =\tfrac 13 k + \tfrac 23 p + \tfrac 23 z$.
Recall that $c( P )$ denotes the cost of the 1-paths, and the total cost of $T$ at the end of the algorithm is at most 
$c(P)+\tfrac 13 k -\tfrac 13p-\tfrac 13z$. Since at most $\tfrac 13 k + \tfrac 23 p + \tfrac 23 z$ nodes have degree one in $T$, we can extend $T$ into a tour of cost at most $c(P) + \tfrac 23 k +\tfrac 13 p +\tfrac 13 z$.

The cost of the solution $x$ can be expressed as $c(P) + \tfrac 12 k + \tfrac 12 p$.
Note that each 1-path connects two cycle nodes, hence $c(P)\ge \tfrac 12 k$.
Moreover, an eared path $A$ is incident to one (if $|A|=1$) or two (if $|A|>1$) 1-paths of length two, since the support graph of $x$ is simple.
Therefore we can lower bound $c( P )$ by $\tfrac 12 k + z$.
Therefore, $\tfrac 76\left(c(P) + \tfrac 12 k + \tfrac 12 p\right) \ge c(P) + \tfrac 1{12}k + \tfrac 16 z+\tfrac 7{12}k +\tfrac 7{12} p
\ge c(P)+\tfrac 23k +\tfrac 13 z + \tfrac 13 p$, where $p \geq z$ is used in the last inequality.
\end{proof}

\begin{proofof}{Lemma~\ref{lem:third}}
The basic idea behind the proof of the lemma is the following: We go through the edges of $A$ in order, and try to add them to $T$ if this does not create a cycle or node of degree three in $T$. If we cannot add an edge, we simply skip the edge and continue to the next edge.  Since the edges in $T$ form a collection of disjoint paths and each node in $A$ has degree one in $T$, we can always add either the first edge or the second edge of $A$: if the first edge cannot be added, then adding it to $T$ must create a cycle, and since the edges in $T$ form a collection of node disjoint paths, adding the second edge of the path or cycle to $T$ cannot create a cycle.
Similarly, we need to skip at most two edges between two edges that are successfully added to $T$: first, an edge is skipped because otherwise we create a node of degree three in $T$, and if a second edge is skipped, then this must be because adding that edge to $T$ would create a cycle. But then, adding the next edge on the path cannot create a cycle in $T$.

To lower bound the number of edges from we can add from each path or cycle $A$, we partition the edges into groups of two or three consecutive edges.
For a path $A$, the first group contains the first two edges, and each subsequent group contains the next three edges. The final group contains the last zero, one or two edges of the path. For each group except the last group, at least one edge is added to $T$. 
Hence, we can conclude that we can add at least $(|A|-4)/3$ from the groups of size three, and 1 for the first group, for a total of $(|A|-1)/3$ edges, where $|A|$ denotes the number of edges in $A$.
For a cycle $A$, we need to be slightly more careful, since the argument that we can add at least one edge from the last group of size three does not hold if the very first edge was added to $T$ (since it may be the case that the first and third edge of the group cannot be added without creating a node of degree three, and the second edge of the group cannot be added without creating a cycle). Therefore, we let the first group contain two consecutive edges, where the {\it second} edge is the edge that was the first to be added to $T$. By the same argument as for the path, we can thus conclude that we can add at least $(|A|-1)/3$ edges.

We now show that by being a little more careful, we can in fact add $|A|/3$ edges if $A$ is a cycle.
Note that the number of nodes in $A$ is odd, and hence there must be some $j$ such that the path in $T$ that starts in $u_j$ ends in some node $v\not\in A$.
We claim that if we consider the edges in $A$ starting with either edge $\{u_{j-1}, u_j\}$ or edge $\{u_j, u_{j+1}\}$, we are guaranteed that for at least one of these starting points, we can add both the first and the third edge to $T$.

Clearly, neither $\{u_{j-1}, u_j\}$ nor $\{u_j, u_{j+1}\}$ can create a cycle if we add it to $T$.
So suppose that $T\cup \{u_{j-1},u_{j}\} \cup \{u_{j+1}, u_{j+2}\}$ contains a cycle. This cycle does not contain the node $u_j$, because the path in $T$ that starts in $u_j$ ends in some node $v\not\in C$. Hence $T$ contains a path that starts in $u_{j+1}$ and ends in $u_{j+2}$. But then $T \cup \{u_{j},u_{j+1}\} \cup \{u_{j+2}, u_{j+3}\}$ does not have a cycle, since if it did, $T$ must have a path starting in $u_{j+2}$ and ending in $u_{j+3}$ which is only possible if $u_{j+1}=u_{j+3}$. Since the number of nodes in $A$ is at least three, this is not possible.
\end{proofof}

We remark that the ratio of $\tfrac76$ in Theorem~\ref{thm:tour76} is achieved if every 1-path contains just one edge of cost 1, and all cycle edges have cost 1. However, in such a case, we could find another optimal F2M solution of the same cost, which has fewer cycle edges:
 If we have a 1-path of cost 1 with endpoints in two different odd cycles of edges with $x(e)=\tfrac 12$, we can obtain the alternative solution by removing the 1-path, and increasing the $x$-value on the four cycle edges incident on its endpoints to 1, and then alternating between setting the $x$-value to 0 and 1 around the cycles. Now, since the cycles are odd, the degree constraints are again satisfied. The objective value does not increase because we only change the $x$-value on edges of cost 1. For a 1-path of cost 1 with endpoints in the same odd cycle, the cycle gives us two paths between the endpoints, one of odd length and one of even length. We can alternate increasing and decreasing the $x$-value by $\tfrac12$ on the odd-length path and finally decrease the $x$-value for the 1-path to $\tfrac12$, to obtain a new F2M solution of the same cost with fewer cycle edges. We note that these modifications may increase the number of components of the F2M solution.

This motivates the following definition. We call an F2M solution {\it canonical}, if all edges in the support have cost 1 and all 1-paths contain at least two edges. If a canonical F2M solution is connected, we can improve the analysis in Theorem~\ref{thm:tour76} to show the following.
\begin{theorem}\label{thm:tour109}
There exists a tour of cost at most $\tfrac {10}9$ times the cost of a  connected canonical F2M solution if $c(i,j)\in \{1,2\}$ for all $i,j$.
\end{theorem}
\begin{proof}
We adapt the final paragraph of the proof of Theorem~\ref{thm:tour76}. As before, the cost of the tour is at most $c(P)+\tfrac23k + \tfrac13p+\tfrac13z$. However, since all cycle edges have cost 1, $p=0$ and $z=0$. The cost of the tour is thus at most $c(P)+\tfrac23k$.

The cost of the F2M solution is $c(P)+\tfrac12k$. Since each cycle node is the endpoint of a 1-path and vice versa, the number of 1-paths is $k/2$. By the fact that $x$ is canonical, each of these 1-paths has cost at least two, so we get that $c( P ) \geq k$.
The proof is concluded by noting that then $\tfrac{10}9\left(c(P)+\tfrac 12k\right) \geq c(P)+ \tfrac 19 k + \tfrac{10}9\cdot\tfrac 12 k  = c(P)+\tfrac23k$.
\end{proof}

\section{An upper bound of $\tfrac{5}{4}$ on the integrality gap}\label{sec:better}
We now show how to use the results in the previous two sections to obtain an upper bound of $\tfrac{5}{4}$ on the integrality gap for the general case.
In addition, we show that if all edges in the support of the optimal subtour LP solution have cost 1, then the integrality gap is at most $\tfrac{26}{21}$.

We will bound the integrality gap of the solution obtained by the Papadimitriou-Yannakakis algorithm, by (i) bounding the difference between the cost of the 2M and the subtour LP, and (ii) bounding the difference between the 2M solution and the tour constructed from it by the Papadimitriou-Yannakakis algorithm.

As in the Papadimitriou-Yannakakis algorithm described in Section~\ref{sec:12gap}, we call a cycle in a 2M a ``pure'' cycle if all its edges have cost 1, and a ``non-pure'' cycle otherwise. 
The idea behind this section is to show that the quantity in (i) can be ``charged'' to the nodes in the non-pure cycle only, and that the quantity in (ii) can be ``charged'' mainly to the nodes in the pure cycles.

We first state the following lemma, which formalizes the second statement. 
\begin{lemma}\label{lem:pure-nonpure}
If $OPT(\SUBT) < n+1$, then the difference between the cost of the 2M used and the tour constructed by the Papadimitriou-Yannakakis algorithm can be upper bounded by $\alpha n_\text{pure} + \beta (n_\text{non-pure}-\ell)$, where
$n_\text{pure}$ is the number of nodes in pure cycles in the 2M, $n_\text{non-pure}$ is the number of nodes in the non-pure cycle, and $\ell$ is the number of edges of cost 2 in the non-pure cycle, 
for any values of $\alpha, \beta$ so that $9\alpha \ge 2$ and $3\alpha + 2\beta \ge 1$.
\end{lemma}
Note that Lemma~\ref{lem:PYsubt} and the assumption that $OPT(\SUBT)<n+1$ imply that the Papadimitriou-Yannakakis algorithm finds a bipartite matching that matches all the pure cycles.
A careful look at the analysis of Papadimitriou and Yannakakis~\cite{PapadimitriouY93} then shows that their algorithm finds a tour which satisfies the lemma. The details basically follow the analysis of Papadimitriou and Yannakakis, and are therefore postponed to Appendix~\ref{app:purenonpure}.

The key observation in this section is that we can indeed restrict our attention to instances with $OPT(\SUBT) < n+1$, the requirement of Lemma~\ref{lem:pure-nonpure}.
\begin{lemma}\label{lem:worstcase}
The worst-case integrality gap is attained on an instance with subtour LP value less than $n+1$, where $n$ is the number of nodes in the instance.
\end{lemma}
\iftoggle{abs}{The idea behind the proof 
is that, if $\lfloor OPT(\SUBT)\rfloor= n+k$, then the total $x$-value on edges with cost 2 is at least $k$. We can add $k$ nodes and for each new node, add edges of cost 1 to each existing node. We obtain a feasible subtour solution for the new instance with the same cost as the solution for the original instance, by rerouting one unit of flow from edges with cost 2 to go through each new node. Also, the cost of the optimal tour on the new instance is at least the cost of the optimal tour on the original instance, and hence, the integrality gap of the new instance is at least the integrality gap of the original instance.}{
\begin{proof}
Consider an instance $I$ on $n$ nodes for which the ratio between the length of the optimal tour and the subtour LP value is $\gamma$, and suppose $OPT(SUBT) = n+k$ for some $k\ge 1$.
We construct an instance $I'$ with $n'=n+1$ nodes, for which the ratio between the length of the optimal tour and the subtour LP value is at least $\gamma$, and the optimal value of the subtour LP is at most $n+k=n'+k-1$. Repeatedly applying this procedure proves the lemma.

If $OPT(SUBT)=n+k$, then the subtour LP solution on $I$ has a total $x$-value of $k$ on edges of cost 2, since the objective value is equal to $\sum_{e\in E} x(e) + \sum_{e\in E: c(e)=2} x(e)$, and $\sum_{e\in E} x(e)=\tfrac 12 \sum_{v\in V}\sum_{e\in \delta(v)} x(e)=n$.
We fix an optimal subtour solution $x$, and we construct $I'$ from $I$, by adding one node $i$, and adding edges $(i,j)$ of cost 1, for every $j$ in $I$ that is incident on an edge $e$ with $c(e)=2$ and $x(e)>0$.
All other edges incident on $i$ get cost 2.
Note that the optimal tour on $I'$ has length at least the length of the optimal tour on $I$, since we can take a tour on $I'$ and shortcut $i$ to obtain a tour on $I$. On the other hand, we can use $x$ to define a feasible solution on $I'$, by ``rerouting'' one unit in total from edges $e=(j,k)$ with $c(e)=2$ to the edges $(j,i)$ and $(i,k)$. Since the cost of this solution on $I'$ is the same as the cost of $x$, the ratio between the length of the optimal tour and the subtour LP value has not decreased.\mqed\end{proof}}
\begin{remark}
We note that the proof of Lemma~\ref{lem:worstcase} implies that to compute integrality gaps or approximation guarantees, we may assume without loss of generality that an instance has an optimal subtour LP value of at most $n+1$, where $n$ is the number of nodes in the instance. If this does not hold, we may add nodes as in the proof of Lemma~\ref{lem:worstcase} without increasing $OPT(\SUBT)$, and a tour of cost $C$ on the extended instance can be shortcut to a tour on the original instance of cost at most $C$.
\end{remark}

\begin{theorem}\label{thm:final}
The integrality gap of the subtour LP is at most $\tfrac{5}{4}$ for the 1,2-TSP, and it is at most $\tfrac{26}{21}$ for 1,2-TSP instances for which $OPT(\SUBT)<n+\tfrac12$, where $n$ is the number of nodes in the instance.
\end{theorem}
\begin{proof}
By Lemma~\ref{lem:worstcase}, we can assume without loss of generality that $OPT(\SUBT)$ $< n+1$.
To compute a tour, we first drop the subtour elimination constraints and find an optimal F2M solution. Since the F2M problem is a relaxation of the subtour LP, and it is half-integral, its objective value is either $n+\tfrac 12$ or $n$. 

We first consider the case $OPT(\SUBT)<n+\tfrac12$, in which case the optimal F2M solution has objective value $n$. Since all edges in the support of the F2M solution have cost 1, we may assume by the arguments preceding Theorem~\ref{thm:tour109} that all 1-paths contain at least two edges of cost 1; in other words, we may assume the components of the F2M solution are canonical. By applying Theorem~\ref{thm:tour109} we convert each fractional component of the F2M solution into a cycle on the nodes in the component.

Note that each cycle that is the result of applying Theorem~\ref{thm:tour109} contains at least one edge of cost 2. By the observation of Papadimitriou and Yannakakis~\cite{PapadimitriouY93}, we may merge these into a single non-pure cycle. The integer components of the F2M solution are pure cycles, since the support of the F2M solution only contains edges of cost $1$. We let $n_\text{pure}$ be the number of nodes in the pure cycles (or, equivalently, in the integer components of the F2M solution), and let $n_\text{non-pure}$ be the number of nodes in the non-pure cycle (or, equivalently, the number of nodes in the fractional components of the F2M solution).
Let $\ell$ be the number of cost 2 edges in the computed 2-matching. By Theorem~\ref{thm:tour109}, $\ell \le \tfrac{1}9n_\text{non-pure}$.

If we apply the Papadimitriou-Yannakakis algorithm to this 2-matching, this increases the cost by at most $\alpha n_\text{pure} + \beta (n_\text{non-pure}-\ell)$, 
provided that $9\alpha\ge 2$ and $3\alpha+2\beta\ge1$ by Lemma~\ref{lem:pure-nonpure}. 
Choosing $\alpha=\tfrac 5{21}, \beta=\tfrac17$, we thus find that the total cost of the tour is at most $n+\ell+\tfrac5{21}n_\text{pure}+\tfrac 17 n_\text{non-pure} - \tfrac 17 \ell\le n+\tfrac 5{21}n_\text{pure} + (\tfrac 17+\tfrac 67\cdot \tfrac 19) n_\text{non-pure} = (1+\tfrac 5{21})n$, where we used the fact that $\ell\le \tfrac19n_\text{non-pure}$.

If $n+\tfrac12 \le OPT(\SUBT)<n+1$, the optimal F2M solution has cost at most $n+\tfrac12$. We temporarily decrease the cost of the unique cost-2 edge in the F2M to 1, and follow the same procedure as above, to find a 2-matching. Let $n_\text{non-pure}$ be the number of nodes in the non-pure cycle, and note that $n_\text{non-pure}$ is at least 9, since a fractional component of a canonical F2M solution contains at least two odd cycles, containing at least six nodes, and at least three 1-paths, containing at least one additional node each.

Let the cost of this 2-matching (with respect to the true costs) be $n+\ell$, where by Theorem~\ref{thm:tour109}, $\ell-1\le \tfrac 19 n_\text{non-pure}$.
As in the case when $OPT(\SUBT)<n+\tfrac12$, we apply the Papadimitriou-Yannakakis algorithm to this 2-matching, and by Lemma~\ref{lem:pure-nonpure} this increases the cost by at most $\alpha n_\text{pure} + \beta (n_\text{non-pure}-\ell)$. We now choose $\alpha=\tfrac14, \beta=\tfrac18$, to get that the total cost of the tour is at most
$n+\ell+\tfrac14n_\text{pure}+\tfrac18n_\text{non-pure}-\tfrac18\ell= n+\tfrac14n_\text{pure}+\tfrac98n_\text{non-pure}+\tfrac78(\ell-1)+\tfrac78$.
Now, recall that $\ell-1\le \tfrac19n_\text{non-pure}$ and that $n_\text{non-pure}\ge9$, and thus $\tfrac78\le \tfrac58 + \tfrac14\cdot\tfrac19 n_\text{non-pure}$.
Hence, we can upper bound the cost of the tour by $n+\tfrac14n_\text{pure}+ (\tfrac 98 + \tfrac78\cdot \tfrac19 +\tfrac14 \cdot \tfrac19)n_\text{non-pure}+\tfrac 58
=\tfrac54(n+\tfrac12) \le \tfrac54 OPT(\SUBT)$.
\mqed\end{proof}

\begin{remark}
The bound of $\tfrac54$ in Theorem~\ref{thm:final} may be marginally improved by a more careful analysis of small instances. It appears that in order to decrease the bound to $\tfrac{10}9$, or even $\tfrac{11}9$, more substantial new ideas are needed, however. 
\end{remark}
\section{Computational results}
\label{sec:comp}

In the case of the 1,2-TSP, for a fixed $n$ we can generate all instances as follows. For each value of $n$, we first generate all nonisomorphic graphs on $n$ nodes using the software package NAUTY~\cite{McKay81}. We let the cost of edges be one for all edges in $G$ and let the cost of all other edges be two. Then each of the generated graph $G$ gives us an instance of 1,2-TSP problem with $n$ nodes, and this covers all instances of the 1,2-TSP for size $n$ up to isomorphism.

In fact, we can do slightly better by only generating biconnected graphs. We say that a graph $G=(V,E)$ is {\em biconnected} if it is connected and there is no vertex $v \in V$ such that removing $v$ disconnects the graph; such a vertex $v$ is a {\em cut vertex}. It is possible to show that the subtour LP value is at least $n+1$ if $G$ is not biconnected, hence, by Lemma~\ref{lem:worstcase} it suffices to consider biconnected graphs.  However, the proof of Lemma \ref{lem:worstcase} involves adding additional new nodes (perhaps many of them).  Using a similar technique to the one in the proof of Lemma~\ref{lem:worstcase}, one can show that given a graph on $n$ vertices, there is a biconnected graph on at most $n+2$ vertices that has no better ratio of optimal tour to subtour LP value. 
In Appendix~\ref{app:twolemmas} we prove two lemmas that imply the following corollary.

\begin{corollary}\label{cor:12biconn}
Let $G=(V,E)$ be the graph of cost 1 edges in a 1,2-TSP instance.  Then if $G=(V,E)$ is not biconnected, there exists a biconnected $G'=(V',E')$ with $|V'|\le |V|+2$ such that $OPT(G)/\SUBT(G) \leq OPT(G')/\SUBT(G')$.
\end{corollary}

For each instance of size $n$, we solve the subtour LP and the corresponding integer program using CPLEX 12.1~\cite{cplex121} and a Macintosh laptop computer with dual core 2GHz processor and 1GB of memory.
It is known that the integrality gap is 1 for $n \leq 5$, so we only consider problems of size $n \geq 6$.   The results are summarized in Table~\ref{tab:12tab}.
\begin{table}[t]
\begin{center}
\begin{tabular}{|c|c|c|c|c|c|c|c|c|}
\hline
$n$ & 6 & 7 & 8 & 9 & 10 & 11& 12 \\
\hline
Subtour IP/LP ratio & $8/7.5$ &  $8/7.5$ &  $9/8.5$ & $10/9$ & $11/10$  & 12/11 & 13/12 \\
\# graphs & 56 & 468 & 7,123 & 194,066 & 9,743,542 & 900,969,091 & $-$ \\
\hline
\end{tabular}
\end{center}
\caption{The subtour LP integrality gap for 1,2-TSP for $6 \leq n \leq 12$, where the ratio for $6 \leq n \leq 10$ is only on biconnected graphs.  The second row shows the number of nonisomorphic biconnected graphs for $6 \leq n \leq 11$.}
\label{tab:12tab}
\iftoggle{abs}{\vspace*{-\baselineskip}}{}
\end{table}
For $n = 11$, the number of nonisomorphic biconnected graphs is nearly a billion and thus too large to consider, so we turn to another approach.
\iftoggle{abs}{For $n = 11$ and $n=12$, we use the fact that we know a lower bound on the integrality gap of $\alpha_n = \frac{n+1}n$, namely for the instances we obtain by adding two or three additional nodes to one of the 1-paths in the example in Figure~\ref{fig:12tsp}. }
{
\begin{figure}
\begin{center}
\includegraphics[height=.75in]{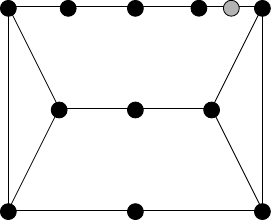}
\end{center}
\caption{Illustration of the instances with integrality gap at least $\frac{12}{11}$ for $n=11$ (without the grey node) and $\frac{13}{12}$ for $n=12$ (with the grey node) for
the 1,2-TSP.  All edges of cost 1 are shown.}\label{fig:badexlarger}
\end{figure}
For $n = 11$ and $n=12$, we use the fact that we know a lower bound on the integrality gap of $\frac{n+1}n$, namely for the instances depicted in Figure~\ref{fig:badexlarger}. The claimed lower bounds on the integrality gap for these instances follow readily from the integrality gap for the example in Figure~\ref{fig:12tsp}.
}
We then check whether this is the worst integrality gap for each vertex of subtour LP. A list of non-isomorphic vertices of the subtour LP is available for $n=6$ to $12$ at Sylvia Boyd's website \url{http://www.site.uottawa.ca/~sylvia/subtourvertices}. In order to check whether the lower bound on the integrality gap is tight, we solve the following integer programming problem for each vertex $x$ of the polytope for $n=11$ and $n=12$, where now the costs $c(e)$ are the decision variables, and $x$ is fixed:
\iftoggle{abs}{
\[\max\{ z - \alpha_n \sum_{e \in E} c(e) x(e): \sum_{e \in T} c(e) \geq z\  \forall \mbox{ tours } T;  c(e) \in  \{1,2\} \  \forall e \in E. \}\]}
{
\lps
& & & \mbox{\sf Max} &z - \alpha_n \sum_{e \in E} c(e) x(e) \\
& \mbox{subject to:} \\
& & & & \sum_{e \in T} c(e) \geq z, & \forall \mbox{ tours } T,\\
& & & & c(e) \in  \{1,2\}, & \forall e \in E.  \elps
}
Note that $\alpha_n$ is the lower bound on the integrality gap for instances of $n$ nodes. If the objective is nonpositive for all of the vertices of the subtour LP, then we know that $\alpha_n$ is the integrality gap for a particular value of $n$.

Since the number of non-isomorphic tours of $n$ nodes is $(n-1)!/2$, the number of constraints is too large for CPLEX for $n =11$  or $12$. We overcome this difficulty by first solving the problem with only tours that have at least $n-1$ edges in the support graph of the vertex $x$, and repeatedly adding additional violated tours. 
We find that the worst case integrality gap for $n=11$ is $\frac{12}{11}$ and for $n=12$ is $\frac{13}{12}$.



\iftoggle{abs}{}{

We can now observe that our overall computation leads to a bound of $\frac{10}9$ on the integrality gap for instances of the 1,2-TSP with $n \leq 12$.  Suppose the worst-case integrality gap for these instances is attained for an instance with $k$ vertices.  If $k \leq 8$, then we know that there is a biconnected graph on at most 10 vertices with no better integrality gap, and we have determined the worst-case ratio for all biconnected graphs on at most 10 vertices.  If $k=9$ or $k=10$, then we know there is a biconnected graph on at most 12 vertices with no better integrality gap, and we have determined the worst-case ratio for all biconnected graphs with at most 10 vertices and all instances with 11 or 12 vertices.  If $k=11$ or $k=12$, then we have determined the worst-case ratio for all instances with 11 or 12 vertices.  Thus for any instance with $n\leq 12$, we have determined that the integrality gap is at most $\frac{10}9$.
}

\section{Conjectures and conclusions}
\label{sec:conc}

As stated in the introduction, we conjecture the following.
\begin{conj} \label{conj:12tsp-gap}
The integrality gap of the subtour LP for the 1,2-TSP is $\tfrac{10}9$.
\end{conj}

Schalekamp, Williamson, and van Zuylen~\cite{SchalekampWvZ11} have conjectured that 
to determine the integrality gap for the subtour LP, we can restrict ourselves to considering instances, which have an optimal solution that is an extreme point of the F2M polytope. 

We have shown in Theorem~\ref{thm:tour76} that if an analogous conjecture is true for 1,2-TSP, then the integrality gap for 1,2-TSP is at most $\tfrac 76$; it would be nice to show that if the analogous conjecture is true for 1,2-TSP then the integrality gap is at most $\tfrac{10}9$.

Finally, we remark that the integrality gap of the linear program obtained by adding the constraints
\[\sum_{e\in\delta(S)\backslash F} x(e) + \sum_{e\in F} (1-x(e)) \ge 1 \quad \forall S \subset V,\ F\subseteq \delta(S),\ |F| \mbox{ odd},\]
to (\SUBT) is at most $\tfrac{11}9$ by (\ref{eq:PY}) and Lemma~\ref{lem:PYsubt}, since the 2M polytope is described by these additional constraints plus the degree constraints.
It is an interesting question whether the analysis of Berman and Karpinski~\cite{BermanK06} can also be expressed in terms of the optimal value of this stronger LP.

\subsection*{Acknowledgements}
We thank Sylvia Boyd for useful and encouraging discussions.
We thank two anonymous referees for helpful comments and suggestions.

\bibliographystyle{spmpsci}      
\bibliography{subtour}

\begin{thebibliography}{10}
\providecommand{\url}[1]{{#1}}
\providecommand{\urlprefix}{URL }
\expandafter\ifx\csname urlstyle\endcsname\relax
  \providecommand{\doi}[1]{DOI~\discretionary{}{}{}#1}\else
  \providecommand{\doi}{DOI~\discretionary{}{}{}\begingroup
  \urlstyle{rm}\Url}\fi

\bibitem{AggarwalGG11}
Aggarwal, N., Garg, N., Gupta, S.: A 4/3-approximation for {TSP} on cubic
  3-edge-connected graphs.
\newblock CoRR \textbf{abs/1101.5586} (2011)

\bibitem{Balinski65}
Balinski, M.L.: Integer programming: Methods, uses, computation.
\newblock Management Science \textbf{12}, 253--313 (1965)

\bibitem{BermanK06}
Berman, P., Karpinski, M.: 8/7-approximation algorithm for (1,2)-{TSP}.
\newblock In: Proceedings of the 17th {ACM-SIAM} Symposium on Discrete
  Algorithms, pp. 641--648 (2006)

\bibitem{BlaeserR05}
Bl\"aser, M., {Shankar Ram}, L.: An improved approximation algorithm for {TSP}
  with distances one and two.
\newblock In: M.~Liskiewicz, R.~Reischuk (eds.) Fundamentals of Computation
  Theory, 15th International Symposium, FCT 2005, \emph{Lecture Notes in
  Computer Science}, vol. 3623, pp. 504--515. Springer (2005)

\bibitem{BoydC11}
Boyd, S., Carr, R.: Finding low cost {TSP} and 2-matching solutions using
  certain half-integer subtour vertices.
\newblock Discrete Optimization \textbf{8}, 525--539 (2011).
\newblock Prior version available at
  \url{http://www.site.uottawa.ca/~sylvia/recentpapers/halftri.pdf}. Accessed
  June 27, 2011

\bibitem{BoydSSS11}
Boyd, S., Sitters, R., {van der Ster}, S., Stougie, L.: {TSP} on cubic and
  subcubic graphs.
\newblock In: O.~G{\"u}nl{\"u}k, G.J. Woeginger (eds.) Integer Programming and
  Combinatorial Optimization, 15th International Conference, IPCO 2011, no.
  6655 in Lecture Notes in Computer Science, pp. 65--77. Springer, Berlin,
  Germany (2011)

\bibitem{Christofides76}
Christofides, N.: Worst case analysis of a new heuristic for the traveling
  salesman problem.
\newblock Report 388, Graduate School of Industrial Administration,
  Carnegie-Mellon University, Pittsburgh, PA (1976)

\bibitem{cplex121}
{IBM} {ILOG} {CPLEX} 12.1 (2009)

\bibitem{DantzigFJ54}
Dantzig, G., Fulkerson, R., Johnson, S.: Solution of a large-scale
  traveling-salesman problem.
\newblock Operations Research \textbf{2}, 393--410 (1954)

\bibitem{GamarnikLS05}
Gamarnik, D., Lewenstein, M., Sviridenko, M.: An improved upper bound for the
  tsp in cubic 3-edge-connected graphs.
\newblock Operations Research Letters \textbf{33}(5), 467--474 (2005)

\bibitem{Goemans95}
Goemans, M.X.: Worst-case comparison of valid inequalities for the {TSP}.
\newblock Mathematical Programming \textbf{69}, 335--349 (1995)

\bibitem{GoemansB90}
Goemans, M.X., Bertsimas, D.J.: Survivable networks, linear programming
  relaxations, and the parsimonious property.
\newblock Mathematical Programming \textbf{60}, 145--166 (1990)

\bibitem{McKay81}
McKay, B.D.: Practical graph isomorphism.
\newblock Congressus Numerantium \textbf{30}, 45--97 (1981)

\bibitem{MnichM13}
Mnich, M., M{\"o}mke, T.: Improved integrality gap upper bounds for {TSP} with
  distances one and two.
\newblock CoRR \textbf{abs/1312.2502} (2013)

\bibitem{MomkeS11}
M{\"o}mke, T., Svensson, O.: Approximating graphic {TSP} by matchings.
\newblock In: Proceedings of the 52th Annual Symposium on Foundations of
  Computer Science, pp. 560--569 (2011)

\bibitem{Mucha11}
Mucha, M.: {13/9}-approximation for graphic {TSP}.
\newblock In: 29th International Symposium on Theoretical Aspects of Computer
  Science, STACS 2012, \emph{LIPIcs}, vol.~14, pp. 30--41 (2012)

\bibitem{OveisGharanSS10}
{Oveis Gharan}, S., Saberi, A., Singh, M.: A randomized rounding approach to
  the traveling salesman problem.
\newblock In: Proceedings of the 52th Annual Symposium on Foundations of
  Computer Science, pp. 550--559 (2011)

\bibitem{PapadimitriouY93}
Papadimitriou, C.H., Yannakakis, M.: The traveling salesman problem with
  distances one and two.
\newblock Mathematics of Operations Research \textbf{18}, 1--11 (1993)

\bibitem{QianSWvZ11}
Qian, J., Schalekamp, F., Williamson, D.P., {van Zuylen}, A.: On the
  integrality gap of the subtour {LP} for the 1,2-{TSP}.
\newblock In: LATIN 2012: Theoretical Informatics, 10th Latin American
  Symposium, \emph{Lecture Notes in Computer Science}, vol. 7256, pp. 606--617
  (2012)

\bibitem{SchalekampWvZ11}
Schalekamp, F., Williamson, D.P., van Zuylen, A.: 2-matchings, the traveling
  salesman problem, and the subtour {LP}: A proof of the {Boyd-Carr}
  conjecture.
\newblock Mathematics of Operations Research \doi{10.1287/moor.2013.0608}.
\newblock To appear. An extended abstract appeared in SODA 2012, pp. 1477-1486

\bibitem{SeboV12}
{Seb\H o}, A., Vygen, J.: Shorter tours by nicer ears: 7/5-approximation for
  graphic {TSP}, 3/2 for the path version, and 4/3 for two-edge-connected
  subgraphs.
\newblock CoRR \textbf{abs/1201.1870} (2012)

\bibitem{ShmoysW90}
Shmoys, D.B., Williamson, D.P.: Analyzing the {H}eld-{K}arp {TSP} bound: A
  monotonicity property with application.
\newblock Information Processing Letters \textbf{35}, 281--285 (1990)

\bibitem{Williamson90}
Williamson, D.P.: Analysis of the {H}eld-{K}arp heuristic for the traveling
  salesman problem.
\newblock Master's thesis, MIT, Cambridge, MA (1990).
\newblock Also appears as Tech Report MIT/LCS/TR-479

\bibitem{Wolsey80}
Wolsey, L.A.: Heuristic analysis, linear programming and branch and bound.
\newblock Mathematical Programming Study \textbf{13}, 121--134 (1980)

\end{thebibliography}
\appendix
\section{Proof of Lemma~\ref{lem:pure-nonpure}} \label{app:purenonpure}
We now prove Lemma~\ref{lem:pure-nonpure} from Section~\ref{sec:better}.
\begin{lemma}
If $OPT(\SUBT) < n+1$, then the difference between the cost of the 2M used and the tour constructed by the Papadimitriou-Yannakakis algorithm can be upper bounded by $\alpha n_\text{pure} + \beta (n_\text{non-pure}-\ell)$, where
$n_\text{pure}$ is the number of nodes in pure cycles in the 2M, $n_\text{non-pure}$ is the number of nodes in the non-pure cycle, and $\ell$ is the number of edges of cost 2 in the non-pure cycle, 
for any values of $\alpha, \beta$ so that $9\alpha \ge 2$ and $3\alpha + 2\beta \ge 1$.
\end{lemma}

\begin{proof}
Recall from Section~\ref{sec:12gap} that the Papadimitriou-Yannakakis algorithm starts by finding a maximum cardinality bipartite matching in a graph which has a node for each pure cycle on one side, and a node for each node in the instance on the other side. There is an edge $(C,i)$ if $i\not\in C$, and there exists some node $j$ in $C$ such that $(i,j)$ is an edge of cost 1.

In Lemma~\ref{lem:PYsubt}, we show that $OPT(\SUBT)\geq n+r$, where $r$ is the number of pure cycles that are not matched in the maximum cardinality bipartite matching.
Hence, the assumption that $OPT(\SUBT)<n+1$ implies that all the pure cycles are matched.
In order to show that this implies the lemma, we will repeat some key parts of the algorithm and analysis of Papadimitriou and Yannakakis.

Consider the directed graph $F=({\cal C},A)$ which has a node for every cycle in the 2M, and an arc $(C,C')$ if the maximum cardinality bipartite matching contains an edge from cycle $C$ to a node $i$ in cycle $C'$. 
Each node in $F$ that corresponds to a pure cycle has outdegree 1, and the non-pure cycle (if it exists) has outdegree 0.
Papadimitriou and Yannakakis show how to find a spanning subgraph of $F'$ of $F$ such that each nontrivial component is an in-tree of depth one or a path of length two. The only possible trivial component is the node that corresponds to the non-pure cycle.
Since the non-pure cycle has outdegree 0, it can only occur in a nontrivial component as the root of an in-tree, or as the endpoint of a path of length two. It turns out that the latter does not happen in the construction described by Papadimitriou and Yannakakis, but even if it did, we could just remove the last edge in the length-two path to obtain one in-tree of depth one and one trivial component containing the non-pure cycle. Hence, we may assume the non-pure cycle only occurs in a nontrivial component as the root of an in-tree.

Papadimitriou and Yannakakis now merge the cycles in one component of $F'$ into a single cycle containing at least one edge of cost 2 as follows:
If the component is an in-tree of depth one, let $C$ be the cycle corresponding to the root, let $C_1, \ldots, C_m$ be the remaining cycles in the component, and let $v_i$ be the node in $C$ such that $(C_i,v_i)$ was in the bipartite matching. 
We consider the nodes in $C$ in clockwise order, starting from a node $v\neq v_i$ for $i=1, \ldots, m$ if such a node exists, and an arbitrary node $v$ otherwise. If we encounter two adjacent nodes $v_i,v_j$ in $\{v_1, \ldots, v_m\}$, then we merge the corresponding cycles $C_i$ and $C_j$ with $C$ according to (a) in Figure~\ref{fig:PY}. Otherwise, if the current node is $v_j$ but its clockwise neighbor is not or if its clockwise neighbor is the first node $v$, then we merge $C_j$ with $C$ as in (b) in Figure~\ref{fig:PY}. Finally, if the component is a path of length two, we merge the three cycles as in (c) in Figure~\ref{fig:PY}. Note that each cycle in the resulting graph contains at least one edge of cost 2, and hence we can find a tour of the same cost by removing the edges of cost 2, and arbitrarily connecting the resulting paths into a tour.

\newlength{\arrowlength}
\settowidth{\arrowlength}{$\rightarrow$}
\newlength{\tw}
\setlength{\tw}{.83\textwidth}
\begin{figure}
\begin{center}
\subfloat[ ]{
\begin{tabular}{cp{\arrowlength}c}
\begin{minipage}[c]{.207\tw}
\includegraphics[width=.207\tw]{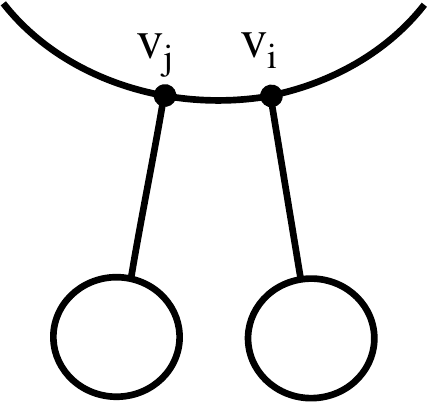}
\end{minipage}
&
\begin{minipage}[c]{\arrowlength} $\rightarrow$ \end{minipage}
&
\begin{minipage}[c]{.207\tw}
\includegraphics[width=.207\tw]{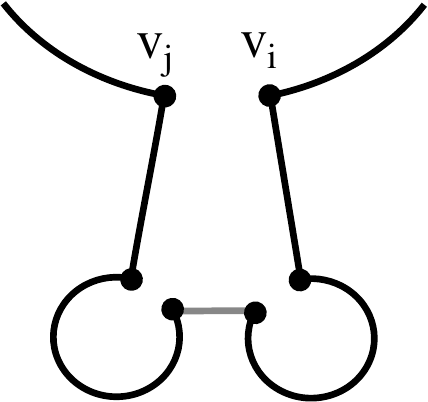}
\end{minipage}
\end{tabular}
}
\qquad
\subfloat[ ]{
\begin{tabular}{cp{\arrowlength}c}
\begin{minipage}[c]{.19\tw}
\includegraphics[width=.19\tw]{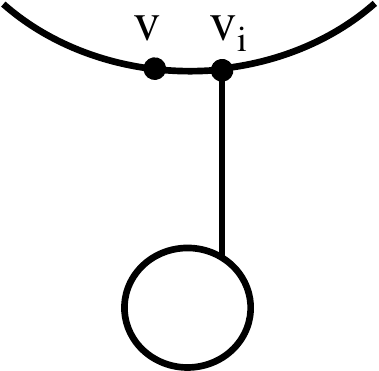}
\end{minipage}
&
\begin{minipage}[c]{\arrowlength} $\rightarrow$ \end{minipage}
&
\begin{minipage}[c]{.19\tw}
\includegraphics[width=.19\tw]{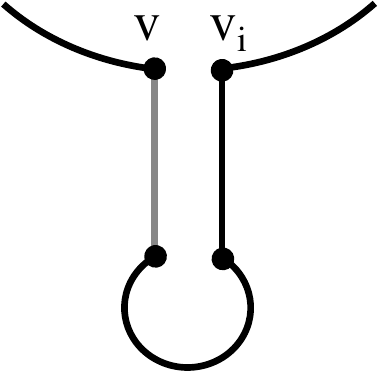}
\end{minipage}
\end{tabular}
}
\\
\subfloat[ ]{
\begin{tabular}{cp{\arrowlength}c}
\begin{minipage}[c]{.26\tw}
\includegraphics[width=.26\tw]{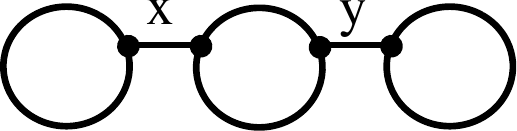}
\end{minipage}
&
\begin{minipage}[c]{\arrowlength} $\rightarrow$ \end{minipage}
&
\begin{minipage}[c]{.26\tw}
\includegraphics[width=.26\tw]{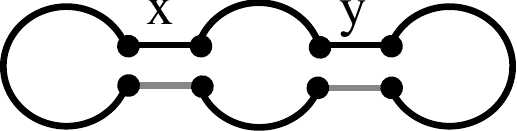}
\end{minipage}
\end{tabular}
}
\end{center}
\caption{The three cases for merging cycles in \cite{PapadimitriouY93} (Figure 2 in \cite{PapadimitriouY93}). The black edges indicate edges of cost 1, and the grey edges (potentially) have cost 2.}\label{fig:PY}
\end{figure}

We now show that the number of edges of cost 2 that are added by merging cycles according to Figure~\ref{fig:PY} can be upper bounded by $\alpha n_\text{pure}+ \beta (n_\text{non-pure}-\ell)$, provided that $\alpha$ and $\beta$ are so that $9\alpha \ge 2$ and $3\alpha + 2\beta \ge 1$.

We say a node is involved in a merging if it is either a node in one of the cycles that are fully drawn in Figure~\ref{fig:PY}, or if it is node $v$ or $v_i$ in subfigure (b). Note that each node is involved in at most one merging. 
Recall that the non-pure cycle can only occur as the root of a 1-tree of depth one or as a trivial component in $F'$, and hence, only the partially drawn cycle in (a) and (b) is (potentially) a non-pure cycle.

We now examine each of the cases (a), (b) and (c) in Figure~\ref{fig:PY} in turn.
In Figure~\ref{fig:PY} (a) one edge of cost 2 is added and we can charge this edge to the (at least) 6 nodes from pure cycles involved in this merging, as long as $6\alpha \geq 1$. This is indeed the case, because we have the stronger requirement that $9\alpha \geq 2$. In (b), again, one edge of cost 2 is added, and we can charge the edge to the (at least) three nodes of the pure cycle involved in the merging and the 2 nodes of the (potentially) non-pure cycle involved in the merging, as long as $3\alpha + 2\beta \ge 1$ (in case the 2 nodes were part of the non-pure cycle), and $5\alpha \geq 1$ (in case the 2 nodes were part of a pure cycle). Finally, in Figure~\ref{fig:PY} (c), two edges of cost 2 are added; we can charge the two edges to the (at least) nine nodes from pure cycles involved in the merging as long as $9 \alpha \geq 2$.

Hence, we have shown that difference in cost between the tour and the 2M can be charged to the nodes, in such a way that each node is charged at most once, and a node in a pure cycle is charged at most $\alpha$ and a node in a non-pure cycle is charged at most $\beta$.

Finally, we remark that a node in a non-pure cycle is charged only in case (b). Now, if $(v_i,v)$ in Figure~\ref{fig:PY}(b) is an edge of cost 2, then there is no need to charge any nodes, since the cost after merging is the same as before the merge. Hence, if we direct all edges of the non-pure cycle in clockwise direction, then the head of the edges of cost 2 is never charged. The total chage to the nodes in the non-pure cycle is therefore at most $\beta(n_\text{non-pure}-\ell)$.
\mqed\end{proof}

\section{Proof of Corollary~\ref{cor:12biconn}} \label{app:twolemmas}
We now show that the worst-case integrality gap for the subtour LP for the 1,2-TSP can be found on graphs of cost 1 edges that are biconnected, as stated in Corollary~\ref{cor:12biconn} in Section~\ref{sec:comp}.  Let $OPT(G)$ and $\SUBT(G)$ be the cost of the optimal tour and the value of the subtour LP (respectively) when $G$ is the graph of cost 1 edges.
We start by proving that the worst case is obtained on a connected graph.
\begin{lemma} \label{lem:12connect}
Let $G=(V,E)$ be the graph of cost 1 edges in a 1,2-TSP instance.  Then if $G=(V,E)$ is not connected, there exists a connected graph $G'=(V',E')$ with $|V'|=|V|+1$ such that $OPT(G)/\SUBT(G) \leq OPT(G')/\SUBT(G')$.
\end{lemma}

\begin{proof}
Suppose $G$ has more than one connected component. We create $G'=(V',E')$ by adding a new vertex $i^*$ to the graph, and adding edges from all $j \in V$ to $i^*$ so that $V' = V \cup \set{i^*}$ and $E' = E \cup \set{(i^*,j):  j \in V}$.  Given a tour of $G'$, we can easily produce a tour of $G$ of no greater cost by shortcutting $i^*$, so that $OPT(G) \leq OPT(G')$.  Let $x$ be an optimal solution to the subtour LP for the graph $G$.  We now define a solution $x'$ for $G'$, where $x'_{ij} = x_{ij}$ if $i$ and $j$ are in the same connected component of $G$, while if $i$ and $j$ are in different connected components of $G$, then we set $x'_{ij} = 0$, $x'_{i^*i}=x_{ij}$, and $x'_{i^*j} = x_{ij}$.  It is easy to see that the cost of $x'$ is the same as that of $x$.  We now argue that there is some solution $x''$ feasible for the subtour LP on $G'$ such that its cost is no greater, so that $\SUBT(G') \leq \SUBT(G)$.  It is clear that the bounds constraints (\ref{boundscons}) are satisfied for $x'$ and the degree constraints (\ref{degreecons}) are satisfied for $x'$ for all $i \in V$; however, the degree constraint for $i^*$ may not be satisfied.  Since for any component $C \subseteq V$ of $G$, $x(\delta(C)) \geq 2$, it is clear that $x'(\delta(i^*)) \geq 2$, but it may be the case that $x'(\delta(i^*)) > 2$.  For the subtour constraints (\ref{subtourcons}), consider any $S \subset V'$, $S \neq \emptyset$, such that $i^* \notin S$.  Then $x'(\delta(S)) \geq x(\delta(S)) \geq 2$, and for any $S \subseteq V'$ with $i^* \in S$, $S \neq \{i^*\}$, $x'(\delta(S)) = x'(\delta(V'-S)) \geq 2$ by the previous argument.  Finally, Goemans and Bertsimas~\cite{GoemansB90} have shown (see also Williamson~\cite{Williamson90}) that if edge costs obey the triangle inequality, and there is some solution $x'$ to the subtour LP in which degree constraints are exceeded but all other constraints are met, then there is another feasible solution $x''$ of no greater cost in which all constraints are satisfied.  Hence we have that $\SUBT(G') \leq \SUBT(G)$. Thus we have that $OPT(G)/\SUBT(G) \leq OPT(G')/\SUBT(G')$.
\mqed\end{proof}

\begin{lemma}\label{lem:12biconn}
Let $G=(V,E)$ be the graph of cost 1 edges in a 1,2-TSP instance.  Then if $G=(V,E)$ is connected but not biconnected, there exists a biconnected $G'=(V',E')$ with $|V'|=|V|+1$ such that $OPT(G)/\SUBT(G) \leq OPT(G')/\SUBT(G')$.
\end{lemma}

\begin{proof}
By hypothesis we assume that the graph $G=(V,E)$ is connected.  Let $i_1, \ldots, i_k$ be all the cut vertices of $G$, and let $C_1,\ldots, C_\ell$ be all the connected components formed when these vertices are removed, so that $C_1,\ldots,C_\ell, \set{i_1},\ldots, \set{i_k}$ form a partition of $V$.  We create a new graph $G'=(V',E')$ by adding a new vertex $i^*$, and adding edges from $i^*$ to each vertex in $C_1 \cup \cdots \cup C_\ell$, so that $V' = V \cup \set{i^*}$ and $E' = E \cup \set{(i^*,j): j \in C_p \mbox{ for some } p}$. We note that $G'$ is biconnected. As before, we have $OPT(G) \leq OPT(G')$ since given a tour of $G'$ we can shortcut $i^*$ to get a tour of $G$.  Let $x$ be an optimal subtour LP solution for graph $G$.  We now argue, as we did in the proof of Lemma~\ref{lem:12connect}, that we can create an $x'$ that costs no more than $x$ such that all the subtour and bounds constraints are obeyed, and all degree constraints are either met or exceeded; this will imply that $\SUBT(G') \leq \SUBT(G)$, and complete the proof.  Suppose without loss of generality that removing cut vertex $i_1$ creates components $C_1$ and $C = C_2 \cup \cdots \cup C_\ell \cup \set{i_2} \cup \cdots \cup \set{i_k}$, so that $C_1$, $\set{i_1}$, and $C$ partition $V$.  We set $x'_{ij} = 0$ and $x'_{i^*i}=x'_{i^*j}=x_{ij}$ if $i \in C_1$ and $j \in C$; $x'_{ij} = x_{ij}$ otherwise.  If $i \in C_1$ and $j \in C$, then $(i,j) \notin E$ since $i_1$ is a cut vertex, so the cost of $x'$ is no more than that of $x$.  The arguments that all constraints are satisfied except for the degree constraint on $i^*$ follow as in the proof of Lemma~\ref{lem:12connect}.  We now must argue that $x'(\delta(i^*)) \geq 2$.  To do this, we show that $\sum_{i \in C_1, j \in C} x_{ij} \geq 1$.  Since $x(\delta(i_1)) = 2$, it must be the case that either $\sum_{j \in C} x_{i_1j} \leq 1$ or $\sum_{j \in C_1} x_{i_1j} \leq 1$; without loss of generality we assume the former is true.  Then since $x(\delta(C_1 \cup \set{i_1})) \geq 2$, and $x(\delta(C_1 \cup \set{i_1})) = \sum_{j \in C} x_{i_1j} + \sum_{i \in C_1, j \in C} x_{ij}$, it follows that $\sum_{i \in C_1, j \in C} x_{ij} \geq 1$, and the proof is complete.
\mqed\end{proof}

\end{document}